\documentclass[runningheads]{llncs}
\usepackage[T1]{fontenc}
\usepackage{todonotes}
\usepackage[english]{babel}
\usepackage{amsmath}
\usepackage{amssymb}
\usepackage{graphicx}
\usepackage[ruled,vlined]{algorithm2e}
\usepackage{booktabs, wrapfig}
\usepackage{hyperref}
\usepackage{xargs}
\usepackage{subfig}
\usetikzlibrary{calc}
\usetikzlibrary{arrows.meta, positioning, shapes.misc, fit}
\usetikzlibrary{arrows.meta, positioning, fit}

\begin{document}
\title{Adversarial Robustness of Time-Series Classification for Crystal Collimator Alignment}

\titlerunning{Robustness of Time-Series Classification for Crystal Collimator Alignment}

\author{Xaver Fink\inst{1,2}, Borja Fernandez Adiego\inst{1}, Daniele Mirarchi\inst{1}, Eloise Matheson\inst{1}, Alvaro Garcia Gonzales\inst{1}, Gianmarco Ricci\inst{3}, Joost-Pieter Katoen\inst{2}}
\authorrunning{X. Fink et al.}
\institute{CERN, Geneva, Switzerland\and
RWTH Aachen University, Aachen, Germany\and DESY, Hamburg, Germany}
\maketitle            
\begin{abstract} 
In this paper, we \emph{analyze and improve the adversarial robustness} of a convolutional neural network (CNN) that assists crystal-collimator alignment at CERN's Large Hadron Collider (LHC) by classifying a beam-loss monitor (BLM) time series during crystal rotation. 
We formalize a local robustness property for this classifier under an adversarial threat model based on real-world plausibility. Building on established \emph{parameterized input-transformation} patterns used for transformation- and semantic-perturbation robustness, we instantiate a preprocessing-aware wrapper for our deployed time-series pipeline: we encode time-series normalization, padding constraints, and structured perturbations as a lightweight differentiable wrapper in front of the CNN, so that existing gradient-based robustness frameworks can operate on the deployed pipeline. For formal verification, data-dependent preprocessing such as per-window \(z\)-normalization introduces nonlinear operators that require verifier-specific abstractions. We therefore focus on attack-based robustness estimates and pipeline-checked validity by benchmarking robustness with the frameworks Foolbox and ART. 
Adversarial fine-tuning of the resulting CNN improves robust accuracy by up to 18.6\,\% without degrading clean accuracy. Finally, we extend robustness on time-series data beyond single windows to \emph{sequence-level robustness} for sliding-window classification, introduce adversarial sequences as counterexamples to a temporal robustness requirement over full scans, and observe attack-induced misclassifications that persist across adjacent windows.

\keywords{Adversarial robustness \and Time series classification \and Time series robustness \and Neural network safety \and Adversarial training.}
\end{abstract}

\section{Introduction}\label{sec:introduction}
CERN operates the world’s largest particle-accelerator complex, including the 27\,km Large Hadron Collider (LHC), which accelerates proton or heavy-ion beams in opposite directions and collides them in four major experiments. In 2026, the LHC will stop for three years (Long Shutdown 3) to be upgraded to the High-Luminosity LHC (HL-LHC), which will significantly increase the number of collisions and the demands on beam protection.

One of the upgrades are \emph{crystal collimators}~\cite{redaelli_crystal_2025} -- a novel collimation technology for heavy-ion beams, first tested on the LHC in 2015 and operational since 2023. These devices use bent crystals to more efficiently and accurately deflect stray particles (halo) away from the main beam into absorbers by channeling the particles through the structured crystal lattice. This is a highly critical task: without proper deflection, halo particles will degrade the accelerator through uncontrolled radiation damage, increase background noise in detectors, or even cause a quench (loss of superconductivity resulting in rapid heating) in one of the LHC magnets. Stray particles are only channeled by the crystal lattice within a small ``critical angle'', precise alignment of the crystal collimators is therefore critical. 
Machine-learning (ML) algorithms are increasingly used to support and automate this alignment: a convolutional neural network (CNN) classifies beam-loss monitor (BLM) time series windows during crystal rotations to assist human operators in finding channeling peaks that can then be optimized numerically~\cite{holzer_beam_2005,ricci_machine_2024}. Looking ahead to future machines (e.g., the Future Circular Collider \cite{benedikt_future_2025}), the push towards greater automation necessitates trusted automation where human operators cannot be the sole bottleneck. One of the safety concerns is the CNN's behavior under expected noisy conditions -- a robustness property that is thus imperative to evaluate and quantify.

Given that, the core problem we address is ensuring the \emph{adversarial robustness of this availability-critical time-series classifier under realistic perturbations}, such as background radiation, process-induced noise, and sensor fluctuations. Formally, we are dealing with properties of the following form: Given an input \(x\) \textit{(the input signal)}, a processing function \(\omega(x)\) \textit{(the deployed preprocessing pipeline)}, and a perturbation set \(\Delta(x)\) \textit{(the physically plausible noise)}, assess a local robustness property of the form \textit{``for all perturbations~\(\delta \in \Delta\), the CNN \(f(\cdot)\) predicts the same label for perturbed input \(x + \delta\)''}. The perturbation set is commonly defined with a vector-norm \(p\) and a threshold \(\varepsilon\) which yields the following problem formulation:

\[ \forall \delta \in \Delta(x), \,\text{s.t.} \quad \| \delta \|_p \leq \varepsilon: \quad \quad \arg\max f(\omega(x)) = \arg\max f(\omega(x+\delta)).\]

Adversarial examples \cite{szegedy_intriguing_2014,goodfellow_explaining_2015} are counterexamples to such adversarial robustness properties and can be used as a robustness measure over a population of samples. They have been demonstrated across several data domains~\cite{zugner_adversarial_2019,carlini_audio_2018,eykholt_robust_2018}, including time-series forecasting or classification ~\cite{carlini_audio_2018,li_multivariate_2020,mode_adversarial_2020,ding_black-box_2023,wu_small_2022,dix_measuring_2023}. Methods that search for adversarial examples are called adversarial attacks. However, existing adversarial attacks against robustness do not directly apply to our problem. Naive application of common adversarial attacks ignores specifics of time-series preprocessing, and structured threat models that go beyond standard \(L_p\)-balls are required to cover BLM data channel dependencies. Without careful consideration this leads to invalid robustness estimates. 

Furthermore, since the model operates on a sliding window, the temporal classification trace requires consideration. For this, we extend the concept of robustness from single windows to sequences over full scans. We introduce adversarial sequences as chains of consistent adversarial examples across consecutive windows, serving as counterexamples to operator requirements for stable classification during crystal orientation scans. This approach links to runtime monitoring of cyber-physical systems where properties are verified over sliding time horizons to ensure the stability of safety-critical decisions \cite{bartocci_specification-based_2018}.

Our study yields key insights: By reparameterizing preprocessing and the threat model as a differentiable layer of the CNN, we enable existing frameworks for adversarial robustness estimation to operate on the deployed processing pipeline. While the resulting end-to-end graph can be exported in standard formats (e.g., ONNX/VNNLIB used in the VNN-COMP \cite{brix_first_2023}), applying abstraction- or discrete optimization-based verification tools requires that all preprocessing operators are supported via algorithm-specific abstractions (we discuss this limitation in Section~\ref{sec:attack-framework}). Further, we improve robust accuracy by up to 18.6\,\% via adversarial fine-tuning without degrading accuracy in the absence of perturbations, demonstrating the usefulness of adversarial fine-tuning in environments where training data is sparse. We find that both preprocessing-awareness and correct modeling of the threat model are crucial to avoid severe misestimation of robustness \(L_p\)-ball assumptions. A proof-of-concept attack exposes persistent misclassifications across adjacent windows, highlighting challenges for future temporal robustness analysis in safety-critical ML systems (e.g., bounding runs of false positives that could lead to suboptimal alignment). These results indicate the need for formal method techniques in real-world deployments of ML in critical systems, with implications for automated control in large-scale infrastructures like future particle colliders \cite{benedikt_future_2025}.

The main contributions of this paper are:
\begin{itemize}
  \item We define a threat model and robustness measure for time-series classifiers that take data preprocessing into account, and apply it to a crystal alignment CNN at CERN. We extend this to robustness over entire sequences.
  \item Building on the idea of a parameterized input-transformation layer to express structured threat models (e.g., Semantify-NN \cite{mohapatra_towards_2020}), we design a lightweight wrapper for time-series pipelines that models differentiable preprocessing and structured perturbations, enabling pipeline-aware adversarial attacks.
  \item We evaluate the method on the deployed classifier with standard robustness frameworks (ART and Foolbox) and show that adversarial fine-tuning improves robustness while maintaining accuracy on clean data.
  \item Finally, we analyze adversarial sequences using a proof-of-concept attack method, demonstrating misclassifications that persist across consecutive time windows.
\end{itemize}

Section~\ref{sec:background} provides background information and the problem statement. Section~\ref{sec:threat-model} defines the threat model under the adversarial robustness framework and extends it to adversarial sequences. Section~\ref{sec:attack-framework} proposes a reparameterization framework for enabling gradient-based robustness analysis under time series preprocessing and our threat model, and describes a proof-of-concept algorithm for the search for adversarial sequences. An experimental evaluation of tools on adversarial example generation, adversarial training, and adversarial sequence generation is presented in Section~\ref{sec:experiments}. Section~\ref{sec:related-work} surveys related work. Finally, Section~\ref{sec:conclusion} concludes.

\section{Background and Problem Formulation}\label{sec:background}
\subsection{Crystal collimation for the HL-LHC program}

Crystal alignment is a critical procedure performed during the commissioning of crystal collimators in the LHC. The crystal collimators are part of the collimation system that protects superconducting magnets and experiments by intercepting the beam halo before it can cause damage. The crystal steers (or ``channels'') the full beam halo at its position into a particle absorber, visualized in Figure \ref{fig:crystal_collimation}. The goal of alignment is to identify the optimal angular orientation of the crystal lattice to ensure stable beam-halo redirection \cite{redaelli_crystal_2025}. To align the crystal, a goniometer rotates the crystal lattice through a range of angles and records the feedback of two beam loss monitors. The alignment decision is based on the resulting two-channel time series. These two channels provide the necessary evidence to verify successful alignment:
\begin{itemize}
    \item The Crystal BLM: Located near the crystal, measuring local beam losses. When channeling occurs, losses at this location typically decrease as particles are ``trapped'' and moved away.
    \item The Absorber BLM: Located downstream at the particle absorber. When channeling is successful, this channel shows a sharp peak in beam loss as the redirected halo hits the absorber.
\end{itemize}

\begin{figure}[t]
    \centering
    \includegraphics[width=0.5\linewidth]{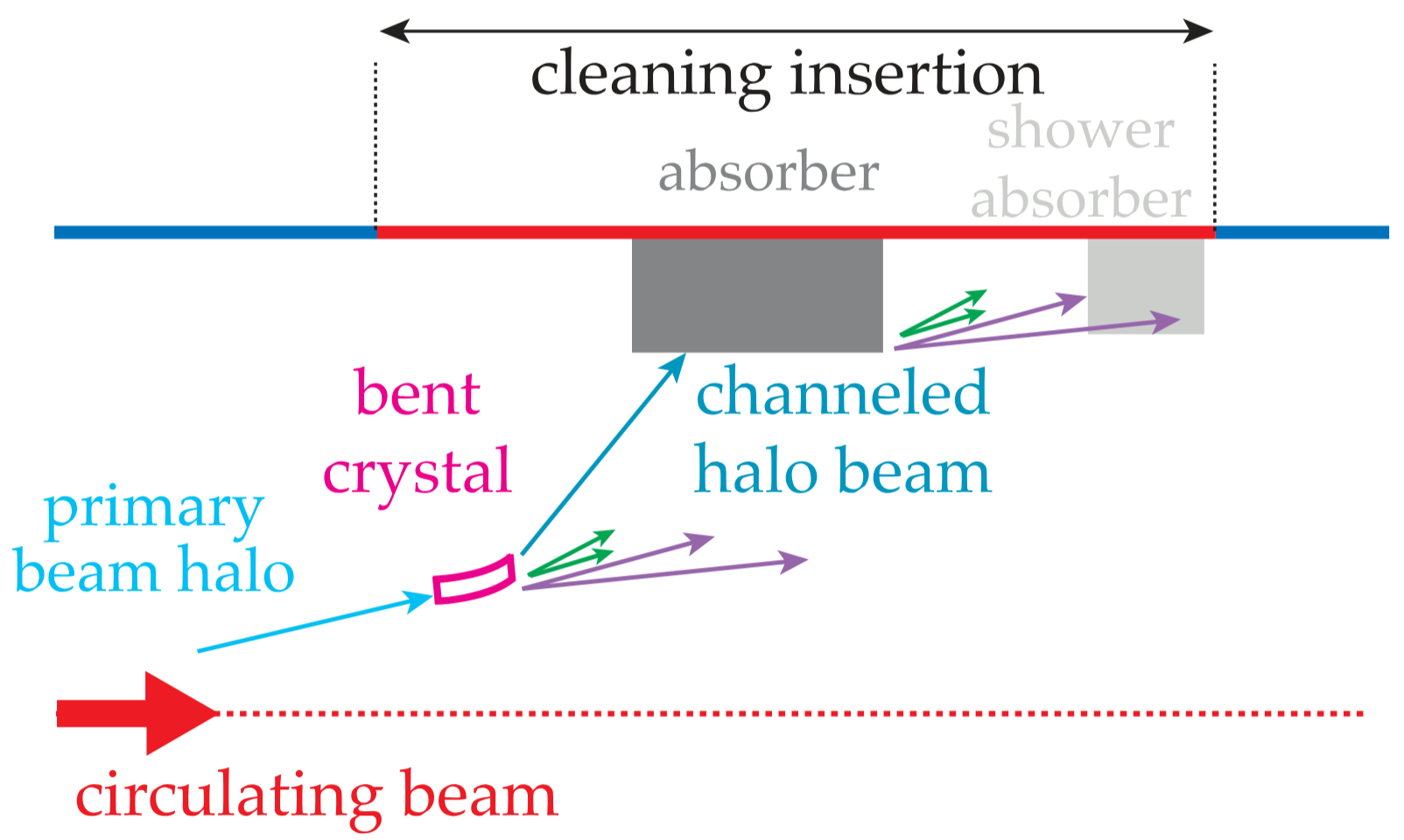}
    \caption{The principle behind the crystal collimation scheme (visualized for betatron cleaning
insertion), showing how the halo particles of the circulating beam are channeled onto the absorber upon correct crystal orientation \cite{redaelli_crystal_2025}.}
    \label{fig:crystal_collimation}
\end{figure}
Alignment and re-optimization are typically performed during commissioning or dedicated machine-development periods with low-intensity pilot beams, where the operational risk is minimized \cite{redaelli_crystal_2025}.
Because dedicated beam time is scarce and valuable, there has been sustained effort to streamline and automate alignment procedures.

\subsection{NN-based classification for crystal-collimator alignment}\label{subsec:nn-background}

To speed up alignment, a semi-automated tool set for BLM-based optimization of crystal orientation is used by LHC collimation experts. In this study, we consider a CNN-based BLM time-series classification during crystal rotations developed by Ricci et al. \cite{ricci_machine_2024}. Their approach classifies BLM signal windows into three possible classes: \textit{no channeling}, \textit{partial well}, and \textit{channeling}. The class \textit{channeling} corresponds to a detection of optimal alignment effects within the window, while \textit{no channeling} corresponds to a lack thereof. The \textit{partial well} class corresponds to rotations where skew planes within the crystal produce a shallower, symmetry-induced trapping effect (which is operationally undesirable). Ricci et al. designed and trained a CNN for classifying sliding windows of the two-channel BLM signal using a dataset comprising 1689 instances labeled from hundreds of angular scans by collimation experts. 

\begin{figure}[t]
    \centering
    \includegraphics[width=0.8\linewidth]{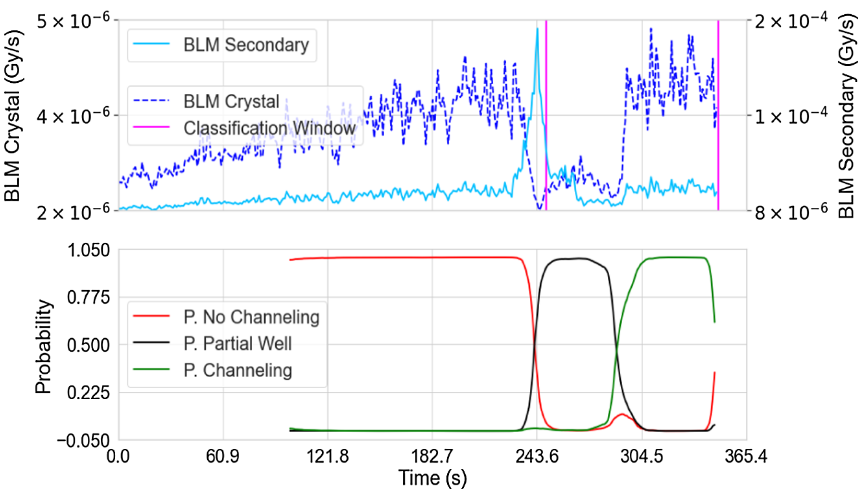}
    \caption{On-the-fly classification during a crystal rotation. At time instant \(i\), the classifier takes the last \(W\) samples (up to \(i\)) of the highlighted window at the right. Top: crystal BLM (dark blue) and secondary BLM (light blue) as a function of time and radiation dose detection; Bottom: predicted class probabilities per window. }
    \label{fig:collimator_window_classification}
\end{figure}

\begin{figure}[t]
    \centering
    \includegraphics[width=\linewidth]{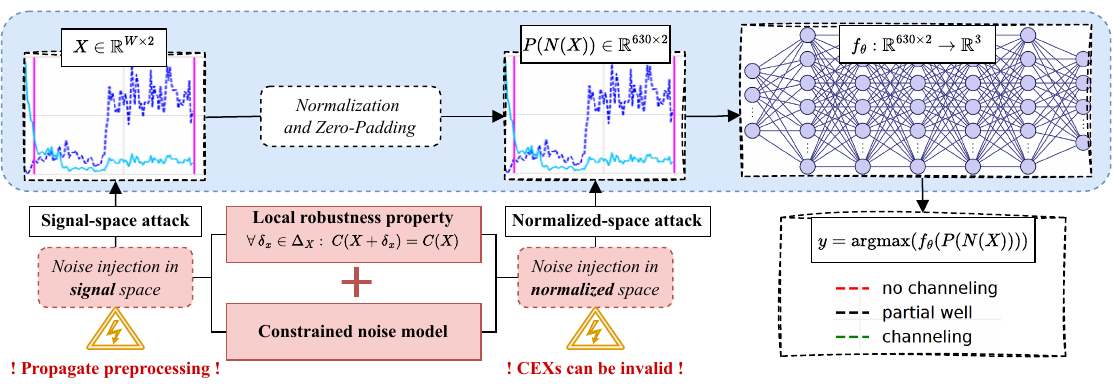}
    \caption{Conceptual overview of the deployed classification pipeline [Section~\ref{subsec:nn-background}] and the robustness problem studied in this paper [Section~\ref{subsec:perturbation}, Section~\ref{sec:threat-model}]. The local robustness property is defined over structured signal-space perturbations \(\delta \in \Delta(X)\) of the input \(X\). Adversarial analysis can target either the signal space (sound but requiring preprocessing to be considered [Section~\ref{sec:attack-framework}]) or the normalized input space (where naive attacks may violate the threat model).}
    \label{fig:cc_pipeline}
\end{figure}

\begin{wraptable}[10]{r}{6.1cm}
\centering
\vspace{-1.5cm}
\caption{CNN architecture (1D CNN) and output shapes \cite{ricci_machine_2024}.}
\label{tab:network_architecture}
\begin{tabular}{ll}
\toprule
Layer (kernel/stride)\ & Output \\
\midrule
Input \((630,2)\) & \((630,2)\) \\
Conv1D \(256\) ch. & \((630,256)\) \\
BatchNorm/ReLU/Dropout & \((630,256)\) \\
Conv1D \(160\) ch. & \((630,160)\) \\
BatchNorm/ReLU/Dropout & \((630,160)\) \\
Global Average Pooling & \((160)\) \\
Dense (logits) & \((3)\) \\
\bottomrule
\end{tabular}
\end{wraptable}

\paragraph{Operational use.}
To find the correct alignment orientation, the collimation expert performs a full scan of the two BLM feedback signals over the possible range of crystal rotation. The feedback from the BLM signals during the rotation can be captured as a two-channel (crystal and secondary BLMs) time series \(X\in\mathbb{R}^{L\times 2}\) of length \(L\) as visualized in Figure \ref{fig:collimator_window_classification} (top).
This time series is analyzed with the CNN using a sliding window approach: Starting at the first data point, windows of size \(W\) are defined as \(X_i\in\mathbb{R}^{W\times 2}\). Each window \(X_i\) is classified by the CNN yielding a plot of class probabilities over time/rotation where each point corresponds to a window as visualized in Figure \ref{fig:collimator_window_classification} (bottom). Human operators then identify the alignment region from the \textit{channeling} probability curve. 
The precise alignment within \(X_i\) is then optimized with numerical methods (not discussed here).

\paragraph{Input and preprocessing.}
Each window \(X_i\in\mathbb{R}^{W\times 2}\) is preprocessed before it is classified by the CNN. Per-sample, per-channel \(z\)-normalization \(N(\cdot)\) is performed as a typical measure for performance and robustness against mean-shifts:
\[N(X_{i}) := \frac{X_{i}-\mu(X_{i})}{\sigma(X_{i})},\]
where \(\mu(\cdot)\) and \(\sigma(\cdot)\) are the sample mean and standard deviation over the \(W\) samples.
Then left-zero padding \(P(\cdot)\) to length \(p\) is added to be invariant to different window sizes, where we set \(p:=630\) based on the maximal window size, i.e.\ \(P(N(X_i))\in\mathbb{R}^{630\times 2}\) is obtained by prefixing \(p-W\) zeros to each channel of \(N(X_i)\).

\paragraph{CNN architecture.}
CNN maps the preprocessed data to three logits, one for each of the three possible classes. The CNN is thus of type \(f_\theta: \mathbb{R}^{630\times 2} \rightarrow \mathbb{R}^3\) where \(\theta\) are the learned parameters. Table~\ref{tab:network_architecture} describes the internal structure of the CNN. 

The final classifier is given as \(C: \mathbb{R}^{W\times 2} \rightarrow \mathbb{R}^3\) with 

\[C(X) := \mathrm{softmax}( f_\theta(P(N(X)))) \quad \text{for window } X.\]

Figure~\ref{fig:cc_pipeline} provides a high-level overview of the deployed classification pipeline, the robustness property of interest, and how it connects with a real-world threat model and adversarial attacks. Formal definitions of the robustness problem, the threat model, and adversarial attacks are given in Sections 3 and 4.

\subsection{Adversarial perturbations during crystal alignment}\label{subsec:perturbation}
As stated in the introduction, we are interested in robustness under perturbations. We outline the perturbation characteristics we aim to capture and the threats they can pose to the crystal alignment.

BLM windows exhibit two main sources of variation: (i) electronic readout and digitization noise (approximately zero-mean, iid), and (ii) beam-related fluctuations that induce correlations across channels (crystal BLM and secondary BLM respond to the same beam conditions). 
Below, we describe how these variations manifest as perturbations on windows and sequences.

\paragraph{Perturbation for a single window.} For a window \(X_i\), the perturbation \(\delta\) must (a) remain small relative to the scale of each channel \(c \in \{0,1\}\) (\(0:\) crystal BLM, \(1:\) secondary BLM), and (b) combine a \emph{common-mode} component (shared across channels) with \emph{channel-specific} noise. A convenient decomposition is
 
\[\delta_c = \underbrace{\delta^{\mathrm{com}}}_{\text{shared across channels}} +\underbrace{\delta^{\mathrm{ind}}_c}_{\text{per-channel}},\]

with both terms bounded in amplitude. For the normalized data, we additionally must respect the zero-padding. Each window is left-zero-padded to a fixed length \(p\) prior to classification by the CNN, thus perturbations must: (c) respect padding (i.e., be exactly zero on padded indices). Perturbations that alter the padded region or unstructured patterns are easy to flag as artifacts and are not operationally relevant.

\paragraph{Perturbation for a classification sequence.}
Operators performing the alignment procedure act on trends over a sequence of windows. Thus, sequence-level perturbations should (a) evolve smoothly over the scan (small changes between consecutive windows in the model probability outputs), (b) remain consistent over the scan (consecutive windows should share perturbation values for common time steps). Rapidly oscillating patterns or spurious adversarial perturbations are unlikely to be considered plausible.

The main risk modes during alignment are: (i) \emph{false channeling} near a suboptimal orientation (prematurely stopping the scan, collimation performance affected), and (ii) \emph{missed channeling} at the true optimum (unnecessary re-scans, efficiency loss). Operators expect perturbations that are padding-aware, small relative to scale, and channel-correlated as well as time-coherent behavior.

We therefore aim at an adversarial formulation that (1) constrains perturbations by small per-channel amplitudes while explicitly allowing a \emph{channel-correlated} plus \emph{channel-specific} structure, and (2) extends from single windows to sequences via a smoothness notion aligned with operator plausibility. Section~\ref{sec:threat-model} formalizes this as our threat model and serves as a basis for the adversarial attack methods in Section~\ref{sec:attack-framework}.

\section{An Operationally Plausible Threat Model}\label{sec:threat-model}
This section incrementally derives our threat model. Let the input signal be \(X \in \mathbb{R}^{L \times 2}\). We first model the electronic readout and digitization noise as per-channel iid Gaussian noise. This can be simplified to a perturbation under an \(L_\infty\)-norm, with per-channel perturbation budgets \(\varepsilon^\mathrm{ind} \in \mathbb{R}^2_{\geq 0}\) derived from the Gaussians' standard deviations (e.g., via a multiple for high-probability coverage). Separate per-channel budgets are necessary as the variance between channels can vary. 
Beam-related fluctuations are modeled as a correlated component shared across channels but scaled to each channel's magnitude. Let \(\delta^\mathrm{com} \in \mathbb{R}^{L}\) denote the normalized common-mode shape with \(\|\delta^\mathrm{com}\|_\infty\leq 1\), and let \(\delta^\mathrm{ind} \in \mathbb{R}^{L \times 2}\) denote the normalized independent components with \(\|\delta^\mathrm{ind}_{\cdot,c}\|_\infty\leq 1\). The adversarial perturbation \(\delta \in \mathbb{R}^{L \times 2}\) is structured as:
\[\delta_{t,c}\ = \varepsilon^\mathrm{com}_{c} \cdot \delta^\mathrm{com}_t + \varepsilon^\mathrm{ind}_{c} \cdot \delta^\mathrm{ind}_{t,c}, \text{ with} \quad c \in \{0,1\},\]
which ensures the total per-channel bound \(\|\delta_{\cdot,c}\|_\infty \leq \varepsilon^\mathrm{com}_{c} + \varepsilon^\mathrm{ind}_{c}\).

\subsubsection{Adversarial examples.}
For a window \(X_i\in\mathbb{R}^{W\times 2}\), let \(\Delta(X_i)\subseteq\mathbb{R}^{W\times 2}\) denote the set of admissible signal-space perturbations under our threat model:
\[\Delta(X_i) = \{ \delta \in \mathbb{R}^{W \times 2} \;|\; \delta \text{ is of the structured form } (\delta_{t,c}) \text{ introduced above}\}.\]
An adversarial example of \(X_i\) is any 
\[
\hat X_i=X_i+\delta \quad\text{with} \quad \delta\in\Delta(X_i), \, \text{s.t.} \quad \hat y(\hat X_i)\neq \hat y(X_i),
\]
where \(\hat y(\cdot)\) is the predicted label of the deployed pipeline (Fig.~\ref{fig:cc_pipeline}).

\subsubsection{Adversarial sequences.} 
For a full scan \(X \in \mathbb{R}^{L \times 2}\), let a sequence \(\pi_X\) of individual window classifications be:
\[\pi_X =\pi_{X_0} = (C(X_0), C(X_1),\dots, C(X_n)),\] 
where each \(X_{i+1}\) is \(X_i\) updated with one new data point in the time series, and \(C(\cdot)\) denotes the (post-softmax) class-probability output. 
Sequences of adversarial examples (dubbed \emph{adversarial sequences}) are now
\[\hat \pi_{X_{i,k}} = (C(\hat X_{i}), C(\hat X_{i+1}),\dots, C(\hat X_{k})),\] 
starting at index \(i\) of the unperturbed sequence. Each \(\hat X_{j} = X_j + \delta_j\) with \(\delta_j \in \Delta(X_j)\), and consecutive perturbations are \emph{consistent}: \(\hat X_{i+1}\) differs from \(\hat X_i\) only in the new data point and its perturbation. A \emph{maximal} adversarial sequence is a sequence of adversarial examples that cannot be extended while preserving adversariality and consistency.
Maximal adversarial sequences of length one are spurious.

State-of-the-art gradient-descent-based adversarial attacks are sound but incomplete \cite{uesato_adversarial_2018}. This is because they search in the gradient-vicinity of the input instead of exhaustively exploring every possible value. Complete verification methods such as interval propagation need to prove the absence of an adversarial example in the neighborhood, which is an NP-complete problem \cite{katz_reluplex_2017}. Certifying that an adversarial sequence is maximal reduces to this problem (Appendix A). Thus, for practical purposes, we study \textit{maximal adversarial sequences under a given attack \(\varphi\)}, which is an adversarial sequence that cannot be extended by attack \(\varphi\).

An adversarial sequence is \emph{smooth} if for each two subsequent adversarial examples \(\hat X_j, \hat X_{j+1}\) it holds:
\[\|C(\hat X_{j}) - C(\hat X_{j+1})\|_1 < \kappa, \quad \text{for some small } \kappa \in \mathbb{R}_{\geq 0}.\]
This enforces a \(\kappa\)-smooth evolution over the scan.

\subsubsection{Threat model in the normalized space.} We can similarly define adversarial noise, adversarial examples, and adversarial sequences on the normalized space (corresponding to the normalized-space attack path in Figure \ref{fig:cc_pipeline}) where the input is \(Z_i := P(N(X_i)) \in \mathbb{R}^{630 \times 2}\) and the classifier is the CNN \(f_\theta: \mathbb{R}^{630 \times 2} \rightarrow \mathbb{R}^3\). The main differences are:

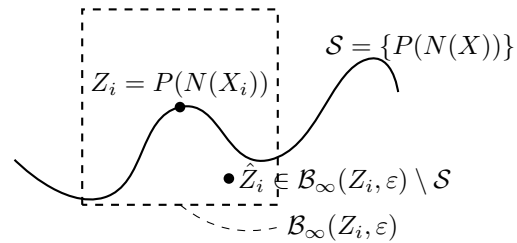
\begin{wrapfigure}{r}{6.1cm}
\vspace{-1cm}
    \centering
    \begin{tikzpicture}[scale=1.0]
  \draw[thick]
    plot[smooth, tension=0.9] coordinates {
      (1.4,2.1) (2.6,1.6) (3.6,2.8) (4.8,2.1) (6.0,3.4) (6.5,3)
    };
  \node[align=left] at (6.8,3.6) {$\mathcal{S}=\{P(N(X))\}$};

  \coordinate (Zi) at (3.6,2.8);
  \fill (Zi) circle (2pt);
  \node[above] at (Zi) {$Z_i=P(N(X_i))$};

  \def\eps{1.3} 
  \draw[thick, dashed]
    ($(Zi)+(-\eps,-\eps)$) rectangle ($(Zi)+(\eps,\eps)$);
  \node[below right] at ($(Zi)+(\eps,-\eps)$)
    {$\mathcal{B}_\infty(Z_i,\varepsilon)$};

   \draw[dashed]  (3.6,2.8-\eps) edge [bend right=25]  (4.9,2.5-\eps);

  \coordinate (Zadv) at (4.25,1.85);
  \fill (Zadv) circle (2pt);
  \node[right] at (Zadv)
    {$\hat Z_i\in \mathcal{B}_\infty(Z_i,\varepsilon)\setminus \mathcal{S}$};
\end{tikzpicture}
    \caption{Infeasible adversarial example \(\hat Z_i\) in the normalized space.}
    \label{fig:ifeasible}
\end{wrapfigure}

\begin{itemize}
    \item a simplified notion of adversarial noise, as one-dimensional perturbation limits \(\varepsilon_\mathrm{\{ind,com\}}\) are sufficient (no scaling by channel-variance).
    \item incorporation of zero-padding.
\end{itemize}

However, adversarial examples in the normalized space may be unsound in the sense that they may violate the intended threat feasibility constraints. Informally, normalization and padding enforce constraints on the feasible set of adversarial examples in the signal space that are not captured by a simple \(L_\infty\) ball. Figure \ref{fig:ifeasible} visualizes this.

\begin{proposition}[Infeasibility due to per-window normalization constraints]\label{prop:ifeasible}
Let \(Z_i := P(N(X_i))\) and let \(\Delta(X_i)\) be the signal-space perturbation set induced by our threat model. The set of feasible normalized inputs is
\[
\mathcal{Z}(X_i) := \{\, P(N(X_i+\delta)) \mid \delta \in \Delta(X_i) \,\}.
\]
Let
\[
\mathcal{B}_\infty(Z_i,\varepsilon) := \{\, \hat Z_i \mid \|\hat Z_i-Z_i\|_\infty \le \varepsilon \,\}\quad \text{for some } \varepsilon > 0
\]
be an adversarial example on \(Z_i\) within the \(L_\infty\)-ball.
Then: for any window \(X_i\) with non-zero per-channel variance and any \(\varepsilon>0\), there exists \(\hat Z\in \mathcal{B}_\infty(Z_i,\varepsilon)\) such that \(\hat Z \neq P(N(X'))\) for all \(X'\); hence \(\hat Z\notin \mathcal{Z}(X_i)\).
\end{proposition}

\begin{proof}[sketch]
For windows with \(\sigma(X_{i,c})>0\), z-normalization ensures that for each channel \(c\) the unpadded normalized vector \(N(X_{i,c})\) has sample mean \(0\) and sample standard deviation \(1\).
Padding only prefixes zeros and does not change these per-window statistics on the unpadded part.
Fix any \(\varepsilon>0\) and modify a single unpadded entry of \(Z_i\) in one channel by \(<\varepsilon\) to obtain \(\hat Z_i\) such that the unpadded part no longer has mean \(0\) (or standard deviation \(1\)).
Then \(\|\hat Z_i-Z_i\|_\infty<\varepsilon\), hence \(\hat Z_i\in\mathcal{B}_\infty(Z_i,\varepsilon)\).
However, \(\hat Z_i\) cannot equal \(P(N(X'))\) for any \(X'\), since \(N(\cdot)\) enforces mean \(0\) and standard deviation \(1\) by definition. Thus \(\hat Z_i\not\in\mathcal{Z}(X_i)\).
\end{proof}

This mismatch implies that perturbing \(Z_i\) directly can violate the intended signal-space constraints; in Section~\ref{sec:attack-framework} we therefore reparameterize the attack in signal space which allows passing through the deployed preprocessing. This issue is not specific to our setting: per-window \(z\)-normalization and padding are standard in time-series classification \cite{bagnall_uea_2018,dau_ucr_2019,wang_time_2017}, so normalized-space \(L_p\) threat models can misrepresent feasible inputs whenever preprocessing depends on the instance.

\section{Threat-Model–Aware Attack Framework}\label{sec:attack-framework}
Figure~\ref{fig:cc_pipeline} summarizes the two natural ways one might apply standard adversarial-attack tooling to our deployed pipeline. A naive application follows the default assumption of frameworks such as ART \cite{nicolae_adversarial_2019} and Foolbox \cite{rauber_foolbox_2020}: one attacks the neural network on its \emph{model input} using an \(L_p\)-bounded perturbation set. The model input is the preprocessed window
\(Z_i := P(N(X_i)) \in \mathbb{R}^{630\times 2}\), and a default robustness evaluation would therefore run (e.g., PGD) directly on \(f_\theta\) with a constraint of the form \(\|\hat Z_i - Z_i\|_p \le \varepsilon\).
This corresponds exactly to the \emph{normalized-space attack} path in Fig.~\ref{fig:cc_pipeline}.

However, Figure~\ref{fig:cc_pipeline} also highlights why this is problematic for our setting. First, our operational threat model is defined in \emph{signal space} and is \emph{structured} (common-mode + per-channel components with per-channel budgets), whereas off-the-shelf tools typically assume a single homogeneous radius \(\varepsilon\). Second, attacking \(Z_i\) directly can yield perturbations that are \emph{infeasible} under the deployed preprocessing: per-window normalization and zero-padding restrict the set of valid preprocessed inputs as shown in Proposition \ref{prop:ifeasible} (and Figure \ref{fig:ifeasible}).
Attempting to ``fix'' candidates post hoc (e.g., by renormalizing or masking padded indices) can destroy adversariality, leading to misleading robustness estimates.

Consequently, we need \emph{signal-space attacks} that optimize over perturbations on \(X_i\) while propagating gradients through the deployed preprocessing (normalization and padding). Existing attack frameworks do not natively support our combination of (i)\emph{ data-dependent preprocessing} and (ii) a structured constraints as threat models and would require code-modifications or adapter-definitions. 

We therefore \emph{instantiate a standard reparameterization/wrapper pattern} for our deployed time-series pipeline: composition of a parameterized input transformation with the classifier (as in EOT-style robustness evaluation \cite{athalye_synthesizing_2018} and semantic/contextual perturbation frameworks such as Semantify-NN and DeepCert \cite{mohapatra_towards_2020,paterson_deepcert_2021}).
Concretely, we encode (a) the structured common\(+\)independent noise model and (b) the preprocessing into an additional differentiable CNN layer, so that standard \(L_\infty\) attacks operate in a \emph{tool-agnostic} way on a bounded auxiliary variable while the resulting perturbations are valid for the deployed pipeline.
Additionally, we show in Section~\ref{subsec:sequence-construction} that the same construction can be used as a building block to generate adversarial sequences.

\paragraph{Scope and limitation.}
The wrapper is \emph{tool-agnostic} in the sense that it targets gradient-based analysis: any attack/optimizer that differentiates through the computation graph can be applied without modifying the attack library.
Formal verification tools are typically more restrictive: per-window \(z\)-normalization computes statistics from the input and contains nonlinear operators (variance, \(\sqrt{\cdot}\), division) that typically require verifier-specific abstract transformers for abstract bound propagation or discrete optimization (MIP, SAT). Our wrapper does not by itself yield such a verifier-agnostic encoding of the full deployed preprocessing when dynamic normalization is included.

\subsection{Attacks over constrained noise}\label{subsec:attacks-noise}
In this subsection, we describe the reparameterization for enforcing the structured noise model from Section~\ref{sec:threat-model} in a generic input space; Section~\ref{subsec:preprocessing-attacks} then composes it with the deployed preprocessing.
For a normalized input \(z \in \mathbb{R}^{630 \times 2}\), valid perturbations combine a channel-correlated (common-mode) component and channel-specific (independent) components. We use the normalized variables directly as attack variables bounded in \([-1,1]\):
\[\delta_{t,c}=\underbrace{\varepsilon^\mathrm{com}_c \cdot u^\mathrm{com}_t}_{\text{shared across channels}}
+\underbrace{\varepsilon^\mathrm{ind}_c \cdot u^\mathrm{ind}_{t,c}}_{\text{per-channel}},
\quad  \text{where }\|u_{\{com,ind\}}\|_{\infty}\le 1,\]
where \(\varepsilon^\mathrm{com}_c, \varepsilon^\mathrm{ind}_c \geq 0\) define per-channel \(L_\infty\) budgets, yielding:
\[\|\delta_{\cdot,c}\|_\infty \leq \varepsilon^\mathrm{com}_c + \varepsilon^\mathrm{ind}_c.\]
We stack the variables as \(u = [u^\mathrm{com}, u^\mathrm{ind}] \in \mathbb{R}^{630 \times 3}\) with \(\|u\|_\infty \leq 1\). This normalization facilitates compatibility with the assumption of having one global \(L_\infty\) attack variable.

To integrate this specification with existing tools (rather than attacking \(z\) directly), we reparameterize the CNN model to take \(u_{com}, u_{ind}\) as input and to optimize over their gradient.

\subsubsection{Reparameterization layer.}
Given the CNN \(f_\theta: \mathbb{R}^{630 \times 2} \rightarrow \mathbb{R}^3\), we reparameterize it under attack perturbation variable \(u \in \mathbb{R}^{630 \times (2+1)}\) by adding a \emph{reparameterization layer} \(U_{z,\varepsilon^\mathrm{com},\varepsilon^\mathrm{ind}}: \mathbb{R}^{630 \times (2+1)} \rightarrow \mathbb{R}^{630 \times 2}\).
\[U_{z,\varepsilon^\mathrm{com},\varepsilon^\mathrm{ind}}(u) = z + \varepsilon^\mathrm{com} \cdot u_\mathrm{com} + \varepsilon^\mathrm{ind} \cdot u_\mathrm{ind}.\]

The wrapped model \(\tilde f(u) = f_\theta(U(u))\) enables standard \(L_\infty\) attacks with on \(u\), enforcing the structured noise on \(z\).

\begin{remark}
   Note that while the added layer \(U(\cdot)\) has to be initialized individually for each input \(z\) and perturbation budget \(\varepsilon\), it can be implemented to operate on sets of inputs by adding an additional dimension which improves computational efficiency. Further, the functional composition of reparameterization-layer construction and adversarial attack can be seen as a generalized algorithm which takes as input (a) the problem instance, (b) the robustness property, and (c) the operational constraints. This composition can be constructed fully algorithmically and does not represent a computational restriction. 
\end{remark}

\subsection{Preprocessing-aware attacks}\label{subsec:preprocessing-attacks}

As shown in Proposition~\ref{prop:ifeasible} and illustrated in Fig.~\ref{fig:ifeasible}, per-window \(z\)-normalization restricts the set of valid preprocessed inputs. Directly perturbing the normalized input \(z\) may yield infeasible \(\hat Z\), while post-hoc renormalization can remove adversariality. Similarly, perturbations in zero-padded regions are invalid, and subsequent masking may remove adversariality.

To address these issues, we perform attacks in the original signal space, incorporating normalization and padding into the computational graph. Define the normalization as \(N(x)=\frac{x-\mu(x)}{\sigma(x)}\) and the padding mask application as \(P(x)=x \cdot m(x_0)\), where \(m(x_0)\) is the zero-padding mask derived from the original input \(x_0\). The perturbation \(\delta x\) is then optimized in the signal space, with gradients propagating through preprocessing, yielding \(f_\theta(P(N(x_0+\delta x)))\).

However, attacks in the signal space must account for \emph{heterogeneous channel variances}: a uniform global bound disproportionately impacts low-variance channels. Thus, normalization-aware attacks require per-channel budgets.

We adapt the reparameterization from Section~\ref{subsec:attacks-noise} to this setting. The reparameterization layer now takes the original input \(x_0\) instead of \(z\), and computes preprocessing via \(\sigma(\cdot)\), \(\mu(\cdot)\), and \(m(\cdot)\), integrating the full classification pipeline into the computational graph for gradient flow. The attack variable is scaled by \(\sigma(x_0)\) to define per-channel budgets in specified in the original signal scale, while the optimizer uses a uniformly bounded variable \(u=(u_\mathrm{com},u_\mathrm{ind})\): 

\[U_{x_0,\varepsilon^\mathrm{com},\varepsilon^\mathrm{ind}}(u) = P(N(\hat x)) = \frac{\hat x - \mu(\hat x)}{\sigma(\hat x)} \cdot m(x_0),  \quad \text{where}\]

\[\hat x = x_0 + \sigma(x_0) \cdot \Big( \varepsilon^\mathrm{com} \cdot u_\mathrm{com} + \varepsilon^\mathrm{ind} \cdot u_\mathrm{ind} \Big).\]

\subsection{Constructing adversarial sequences}\label{subsec:sequence-construction}
We now instantiate the adversarial-sequence notion of Section~\ref{sec:threat-model} by constructing counterexamples to sequence-level robustness over a scan using standard attack tools.
By intuition, adversarial sequences should be found around adversarial examples whenever the adversarial example window already covers the signal parts of highest gradient magnitude. If a new sample is added to the input, its effect on the activation is bounded by the single-point gradient and the gradient effect of window-shifting. If the adversarial examples are caused by local structures of the CNN that are not shift-invariant, shifting can cause the adversarial example to disappear. If they are not, smooth adversarial sequences should emerge naturally from single adversarial examples.

As a proof-of-concept we describe a method for creating adversarial sequences to test this hypothesis. Let \(X \in \mathbb{R}^{L\times 2}\) be a BLM feedback and \(\pi_X\) the corresponding classification sequence and \([a,b] \in \mathbb{R}^{b-a}\) a sequence of classification indices for which the classification should be flipped to a different class. Then we create an adversarial sequence by optimizing a perturbation on \(X\) that jointly maximizes the number of misclassifications for windows \(X_a,\dots,X_b\). This is done by defining a new reparameterization layer \(V_{X}:\mathbb{R}^{L \times 2}\) which can be attacked with Foolbox or ART.

\begin{remark}
    This algorithm maximizes misclassifications but does not guarantee maximality or optimal smoothness; future work could incorporate smoothness constraints directly. In Section~\ref{sec:experiments} we show that smooth adversarial sequences still emerge when applying the algorithm. 
\end{remark}

\section{Threat Model Robustness: Attacks, Defense, Sequences}\label{sec:experiments}
In this section, we (i) measure robust accuracy (RA) on our threat model and compare it to the RA of threat models only partially applying the preprocessing-awareness and noise-correlation methods introduced in Section~\ref{sec:attack-framework}. Then, we (ii) evaluate the effectiveness of adversarial fine-tuning and compare the RA under three perturbation radii. Finally, we (iii) apply the adversarial sequence attack on a BLM sequence as a proof-of-concept. 

\paragraph{Evaluation target.}
We evaluate robustness of the deployed prediction
\[
\hat y(x) := \arg\max C(x) ,
\]
where \(C(x)\) is the deployed pipeline from Section~\ref{subsec:nn-background}.

\paragraph{Attack-based robust accuracy.}
Let \(\mathcal{D}_{\mathrm{test}}\subseteq\mathcal{X}\) be the set of test inputs and \(\Delta(x)\) the admissible perturbation set for \(x\) as defined in Section~\ref{sec:threat-model} (instantiated via \((\varepsilon^\mathrm{com},\varepsilon^\mathrm{ind})\)).
Given an \emph{untargeted} attack procedure \(\mathcal{A}\) that either returns a candidate adversarial example \(\hat x \in x+\Delta(x)\) or returns \(\bot\) (failure), we define
\[
\mathrm{RA}_{\mathcal{A}}
:= \frac{1}{|\mathcal{D}_{\mathrm{test}}|}
\sum_{x\in \mathcal{D}_{\mathrm{test}}} \mathbf{1} \quad \mathrm{iff} \quad \mathcal{A}(x)=\bot.
\]
Intuitively, \(\mathrm{RA}_{\mathcal{A}}\) is the fraction of test inputs for which the attack does not find any admissible \(\hat x \in x+\Delta(x)\) with \(\hat y(\hat x)\neq \hat y(x)\).
Since \(\mathcal{A}\) is incomplete, \(\mathrm{RA}_{\mathcal{A}}\) is a non-certified (optimistic) upper-bound estimate of robustness.

To isolate the impact of missing preprocessing-awareness or threat-model structure, we run attacks under multiple analysis configurations.
Depending on the configuration, the attack optimizes a simplified problem (e.g., without padding mask or with frozen normalization statistics) and therefore reports success in its own optimization space.
We thus report two metrics:
(i) \emph{tool-reported} robust accuracy \(\mathrm{RA}^{\mathrm{tool}}\), as returned by the attack in the configuration's optimization space; and
(ii) \emph{pipeline-checked} robust accuracy \(\mathrm{RA}^{\mathrm{pipe}}\), obtained by reconstructing each candidate into signal space and counting it as adversarial only if it is admissible under the true threat model \(\Delta(\cdot)\) and \(\hat y(\hat x)\neq \hat y(x)\).
In the remainder of this section, we instantiate \(\mathrm{RA}_{\mathcal{A}}\) either as \(\mathrm{RA}^{\mathrm{tool}}\) (evaluated in the configuration's optimization space) or as \(\mathrm{RA}^{\mathrm{pipe}}\) (pipeline-checked).

\subsection{Experimental setup}
\paragraph{Hardware.}
All experiments ran on CERN SWAN (AlmaLinux~9, GCC~11) using 4 CPU cores of an AMD EPYC~7313, 32\,GB RAM, and a 10\,GB A100 MIG slice (CUDA~12.5); Python~3.11.9.

\paragraph{Tooling.}
For adversarial example generation we use the two largest adversarial robustness frameworks for Python: the Adversarial Robustness Toolbox (ART) \cite{nicolae_adversarial_2019} and Foolbox \cite{rauber_foolbox_2020}.

\subsection{Adversarial attacks}
Our goal is to evaluate the classifier under perturbations that are both (i) compatible with the deployed preprocessing (per-sample, per-channel \(z\)-normalization and left padding) and (ii) consistent with the threat model in Section~\ref{sec:threat-model}.

Unless stated otherwise, we use untargeted \(L_\infty\) projected gradient descent (PGD) as the main attack (Foolbox and ART backend) on the reparameterized input \(u\), with a radius \(\varepsilon=1\) (full structured budget) and a short sweep over smaller radii for curves. For baselines that perturb \(z\) directly, we match budgets by setting \(\varepsilon_\mathrm{glob}=\max_c(\varepsilon^\mathrm{com}_c+\varepsilon^\mathrm{ind}_c)\). All attacks use identical PGD hyperparameters (chosen once based on initial experiments) and every returned candidate is pipeline-checked for (i) admissibility under \(\Delta(\cdot)\) (mask and per-channel bounds) and (ii) a prediction change under the deployed pipeline \(C\). Unless stated otherwise, we report \(\mathrm{RA}^{\mathrm{tool}}\) and \(\mathrm{RA}^{\mathrm{pipe}}\) on the test split (B=296) of Ricci et al. \cite{ricci_machine_2024}, where the CNN achieves clean accuracy \(0.939\).

\begin{table}[t]
\centering
\caption{Robust accuracy (RA) on the same test split (B=296) under different configurations. Each cell shows \(\mathrm{RA}^{\mathrm{tool}} \rightarrow \mathrm{RA}^{\mathrm{pipe}}\), i.e., tool-reported robustness in the optimization space followed by pipeline-checked robustness under \(C\) and \(\Delta(\cdot)\).}
\label{tab:expA}
\begin{tabular}{lcc}
\toprule
\textbf{Configuration} & \textbf{Foolbox} & \textbf{ART} \\
\midrule
Baseline
  & \(0.794 \rightarrow 0.794\)
  & \(0.838 \rightarrow 0.838\) \\
\midrule
No-normalization
  & \(0.733 \rightarrow 0.831\)
  & \(0.777 \rightarrow 0.868\) \\
No-padding 
  & \(0.669 \rightarrow 0.899\)
  & \(0.716 \rightarrow 0.919\) \\
\midrule
Naive (mask+renorm)
  & \(0.669 \rightarrow 0.669\)
  & \(0.713 \rightarrow 0.713\) \\
Naive (no mask/renorm)
  & \(0.095 \rightarrow 0.767\)
  & \(0.152 \rightarrow 0.818\) \\
\bottomrule
\end{tabular}
\vspace{0.25em}
\end{table}

\subsubsection{A1: Does preprocessing-awareness matter?}
The first experiment evaluates \emph{whether explicitly modeling per-sample normalization and left-padding }(Section~\ref{sec:attack-framework}) \emph{is necessary for robustness evaluation}.
For each configuration, we report \(\mathrm{RA}^{\mathrm{tool}} \rightarrow \mathrm{RA}^{\mathrm{pipe}}\), i.e., robust accuracy as measured in the attack's optimization space versus after pipeline-checking under the deployed pipeline and the true threat model \(\Delta(\cdot)\).

We compare three configurations: \textit{Baseline} (full wrapper: structured noise + per-sample \(z\)-norm + padding mask),
\textit{No-normalization} (as Baseline but using frozen \((\mu,\sigma)\) from \(x_0\)), and
\textit{No-padding} (as Baseline, but we ignore the left-padding mask (i.e., the attack may perturb padded indices)).

Both ablations remove constraints from the optimization problem, so the attack can exploit artifacts that do not transfer to the deployed pipeline.
Consequently, \(\mathrm{RA}^{\mathrm{tool}}\) can be overly pessimistic compared to \(\mathrm{RA}^{\mathrm{pipe}}\). We verify this by reconstructing each returned candidate \(\hat x\) (if needed) and re-evaluating it under the deployed pipeline \(C\).

We set per-channel budgets \((\varepsilon^\mathrm{com},\varepsilon^\mathrm{ind})=(0.10,0.02)\). These budget values were selected based on an empirical analysis of baseline noise characteristics observed in nominal BLM signal data, representing a conservative upper bound on plausible sensor fluctuations.

\paragraph{Results.}
Table~\ref{tab:expA} confirms that missing preprocessing-awareness can substantially distort robustness estimates:
for \textit{No-normalization} and \textit{No-padding}, the tool-side metric \(\mathrm{RA}^{\mathrm{tool}}\) is much lower than the pipeline-checked metric \(\mathrm{RA}^{\mathrm{pipe}}\) (gaps of \(9\)--\(23\) percentage points), because many candidates either become inadmissible or lose adversariality once evaluated under the deployed preprocessing. This is consistent with the pipeline mismatch illustrated in Figure \ref{fig:cc_pipeline}.
In contrast, the preprocessing-aware baseline shows \(\mathrm{RA}^{\mathrm{tool}}\approx \mathrm{RA}^{\mathrm{pipe}}\), indicating that the optimization problem matches the deployed pipeline.

\subsubsection{A2: Does the threat model matter?}
We now ask \emph{how much the admissible perturbation set itself affects measured robustness}. 
We compare (i) our structured model (common + per-channel noise) to (ii) a global \(L_\infty\) on a normalization-aware model on \(x\) (signal space), and (iii) a global \(L_\infty\) ball on \(z\) (normalized space). For evaluation fairness we set \(\varepsilon_{\mathrm{glob}} := \max_c(\varepsilon^\mathrm{com}_c + \varepsilon^\mathrm{ind}_c)\).

\begin{table}[t]
    \centering
   \caption{Comparison of clean accuracy and three RA configurations under the baseline model and an adversarially trained model. The configurations use the following channel budgets: RA@0.05 \((\varepsilon^\mathrm{com},\varepsilon^\mathrm{ind})=(0.05,0.01)\); RA@0.10 \((\varepsilon^\mathrm{com},\varepsilon^\mathrm{ind})=(0.1,0.02)\); RA@0.20 \((\varepsilon^\mathrm{com},\varepsilon^\mathrm{ind})=(0.2,0.04)\).}
    \begin{tabular}{lccccc}
        \toprule
        \textbf{Model} & \textbf{Clean acc}& \textbf{RA@0.05} & \textbf{RA@0.10} & \textbf{RA@0.20} \\
        \midrule
        Baseline    & 0.939& 0.872 & 0.794 & 0.392 \\
        Adv-trained (fine-tune) & 0.943& 0.919 & 0.862 & 0.578 \\
        \bottomrule
    \end{tabular}
    \label{tab:benchmark_summary}
\end{table}

\paragraph{Results.} Our experimental results (Table \ref{tab:expA}) show that a naive model of the adversarial noise underestimates the robustness of the classifier substantially. Combining the naive threat model with missing preprocessing-awareness yields a gap of \(\approx 70\) percentage points between \(\mathrm{RA}^{\mathrm{tool}}\) in the naive configuration and \(\mathrm{RA}^{\mathrm{pipe}}\) under our baseline threat model, highlighting the sensitivity of robustness conclusions to threat-modeling choices. This result shows that \emph{the adversarial noise model is highly relevant for the evaluation of adversarial robustness on real-world data}.

\subsection{Adversarial defenses}

Following the results on adversarial attacks, we evaluate the effectiveness of adversarial defenses on the crystal collimator, by performing \emph{adversarial training} on the CNN. For that we compare the clean (unperturbed) accuracy and three RA configurations of the baseline model with the adversarially trained model.

\paragraph{Adversarial training (fine-tune).}
Starting from the clean model, we fine-tune the model with on-the-fly PGD adversarial examples. Intuitively, this is intended to smooth the decision boundary and eliminate adversarial examples. For details we refer to \cite{madry_towards_2018,bai_recent_2021}. While not part of this experimental evaluation, we found that adversarial fine-tuning improved RA and general accuracy on our case study compared to full adversarial training.

\paragraph{Results.} The results in Table \ref{tab:benchmark_summary} show that fine-tune adversarial training increases RA by \(+4.7\,\%\) to \(+18.6\,\%\) at the reported radii while improving clean accuracy by \(+0.4\,\%\), suggesting a beneficial data-augmentation effect in our sparse training-data context. 

\begin{figure}[t]
    \centering
    \includegraphics[width=0.8\linewidth]{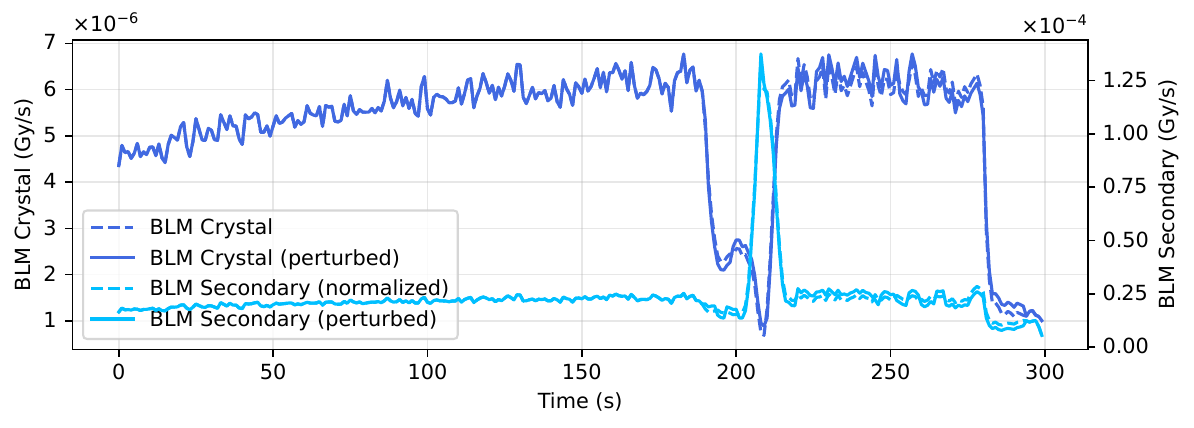}\\
    \vspace{-0.5cm}
    \includegraphics[width=0.8\linewidth]{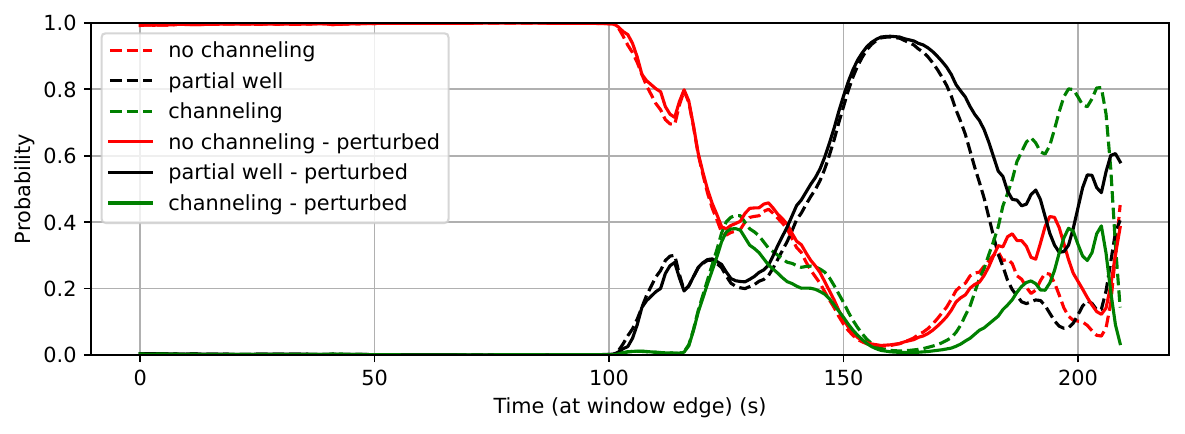}
    \caption{Proof-of-concept adversarial sequence attack. Top: BLM Crystal and Secondary signals (dashed) and their perturbed counterparts (solid). Bottom: class probabilities as the sliding window advances. The perturbation smoothly flips classifications from \textit{channeling} to \textit{partial well} in the targeted region [190, 205].}
    \label{fig:sequence-attack}
\end{figure}

\subsection{Adversarial sequence attacks}
We assess whether adversarial sequences can arise on real BLM data by constructing a single, time-localized perturbation on a representative trace that exhibits all three classes and performing a qualitative evaluation. We target the \emph{missed-channeling} threat model (Section~\ref{subsec:perturbation}), optimizing one perturbation over windows \([190,205]\) to reduce the true class (channeling) confidence. The perturbation is bounded per timestep and channel by the same \(L_\infty\) budget as used for the baseline in the previous experiments.

\paragraph{Results.} Figure~\ref{fig:sequence-attack} shows that a sequence perturbation consistently suppresses channeling confidence in the attacked region, flipping predictions to \emph{partial well}, while leaving the rest of the trace visually and semantically unchanged. This demonstrates the feasibility of adversarial sequence attacks under realistic noise constraints even with a simple proof-of-concept algorithm. Interestingly, the adversarial noise acts as a scaling effect on the probability curves.

\section{Related work}\label{sec:related-work}
At CERN, adversarial training has been explored at the LHC CMS experiment \cite{sarkar_run_2025,malara_exploring_2024}, and neural-network verification was applied to a cooling-tower control network \cite{lopez-miguel_verification_2023}. 

While neural-network verification has seen significant progress through randomized smoothing and deterministic tools such as Marabou and \(\alpha\beta\)-CROWN \cite{cohen_certified_2019,wang_beta-crown_2021,wu_marabou_2024}, applying these tools to real-world pipelines remains a challenge due to non-standard preprocessing. 
By encoding preprocessing as an explicit front-end layer, the end-to-end pipeline can be exported in standard interchange formats used by verification benchmarks (e.g., ONNX/VNNLIB in VNN-COMP \cite{brix_first_2023}). However, whether formal verification is possible depends on the tools' supported operator set and available relaxations. 

Realistic time series or video threat models have been studied in the literature, including smooth or structure-aware perturbations \cite{belkhouja_adversarial_2023,carlini_audio_2018,ding_black-box_2023,han_deep_2020,jiang_black-box_2019}. Our novelty is a sensor-consistent common\(+\)per-channel decomposition and the analysis of window-sequences on time-series data, which also parallels challenges in monitoring mission-critical sensors \cite{bartocci_specification-based_2018}.

Our reparameterization instantiates the pattern of composing a parameterized input transformation with the classifier, as in EOT \cite{athalye_synthesizing_2018} and semantic/contextual perturbation frameworks (Semantify-NN, DeepCert \cite{mohapatra_towards_2020,paterson_deepcert_2021}): we consider a single deterministic, \emph{data-dependent} transform and couple it with a change of variables that enforces structured budgets and padding by construction.

Adversarial sequences have been examined for recurrent models, video recognition, and sequential decision making \cite{papernot_crafting_2016,jiang_black-box_2019}. Our setting differs: RNN sequence attacks assume differentiable closed-loop dynamics; video attacks induce global per-frame pixel changes; and sequential decision making targets trajectory-level manipulations orthogonal to our objective. Our adversarial sequences can be interpreted as counterexamples to Signal Temporal Logic (STL) specifications about stability and persistence of classifications over time, aligning with formal monitoring efforts that seek to detect and bound failure-inducing behaviors in mission-critical infrastructures \cite{bartocci_specification-based_2018}.

\section{Conclusions}\label{sec:conclusion}
We studied the robustness of a crystal-alignment classifier under a channel-correlated model of BLM perturbations. Our robustness evaluation showed that small perturbations can cause substantial adversarial examples for the BLM classifier. These represent real counterexamples that can have impact on CERN operations. We found that we can improve robustness via adversarial fine-tuning, without harming clean classification accuracy. However, despite adversarial fine-tuning, successful attacks remain, indicating that fully autonomous alignment is still challenging.

Our evaluation shows that robustness can be substantially misestimated when any part of the deployed preprocessing pipeline or the intended threat model is omitted. This motivates treating preprocessing as part of the robustness problem: threat-model–aware wrappers make it possible to reuse standard gradient-based attack frameworks while ensuring perturbations remain valid for the deployed pipeline.
The same pattern is broadly applicable to other time-series control and monitoring settings (e.g., different BLM layouts, beam steering, or energy optimization at CERN) by adapting only the padding mask and per-channel budgets. More generally, the results transfer to other domains where global \(L_\infty\) perturbations are an insufficient threat model or complex preprocessing is unavoidable (e.g., medical and industrial-control pipelines).
For formal verification, applicability depends on whether preprocessing can be expressed with verifier-supported operators (typically affine/piecewise-linear) or whether sound, verifier-specific abstractions are available for nonlinear steps.

Future work will refine the noise model, develop verifier-specific abstractions for nonlinear preprocessing, and extend adversarial-sequence evaluation.

\bibliography{references.bib} 

\appendix
\section{Sequence-Level Robustness and Extendability}
This appendix formalizes the notion of extendability of adversarial sequences introduced in Section \ref{sec:threat-model} and clarifies its relationship to local robustness verification.

\subsection{Consistency and Extendability}
Recall that a sliding-window classifier evaluates a sequence of overlapping windows
\[
X_0, X_1, \dots, X_n,
\]
where each window \(X_{i+1}\) is obtained from \(X_i\) by removing the oldest sample and appending one new sample.
In Section~\ref{sec:threat-model} we require \emph{consistency} of adversarial sequences: two consecutive adversarial windows \(\hat X_i\) and \(\hat X_{i+1}\) may only differ in the newly appended sample (subject to the threat model), while all overlapping samples are shared.

\begin{definition}[Extendability]
Let \(\hat X_i\) be an adversarial window at index \(i\).
We say that \(\hat X_i\) is \emph{extendable} if there exists an admissible perturbation of the newly appended sample (under the threat model) such that the resulting window \(\hat X_{i+1}\) is also adversarial.
\end{definition}

This definition makes explicit that once \(\hat X_i\) is fixed the degrees of freedom for the extension are confined to the perturbation of the new sample only.

\subsection{Reduction to Local Robustness}

The following shows that deciding non-extendability reduces to a local robustness (non-existence) question for the next window.

\begin{proposition}[Non-extendability reduces to local robustness]
Fix an adversarial window \(\hat X_i\), assume the consistency rule described above, and let \(\delta_{i+1}\) be an admissible perturbation of only the newly appended sample (under the threat model).
Then \(\hat X_i\) is extendable iff
\[
\exists\, \delta_{i+1}:\quad
\arg\max C(\hat X_{i+1}(\delta_{i+1})) \neq \arg\max C(X_{i+1}).
\]
where \(\hat X_{i+1}(\delta_{i+1})\) denotes the window obtained by appending the perturbed new sample to \(\hat X_i\).
Thus, certifying that \(\hat X_i\) is non-extendable requires proving a universal non-existence claim over the set of admissible perturbations.
\end{proposition}

\begin{proof}[sketch]
By definition of consistency, any extension of \(\hat X_i\) to index \(i{+}1\) must share all overlapping samples with \(\hat X_i\) and can only perturb the newly appended sample.
Thus, an extension exists when there exists an admissible perturbation \(\delta_{i+1}\) such that the resulting window \(\hat X_{i+1}(\delta_{i+1})\) is misclassified.
Non-extendability is therefore equivalent to the absence of such a perturbation, which is a standard robustness non-existence property.
\end{proof}

\begin{remark}
The fact that \(\hat X_i\) is already adversarial can in some cases provide sufficient conditions for persistence of misclassification across windows.
For example, if the classifier logits admit a Lipschitz bound and the adversarial margin at \(\hat X_i\) is sufficiently large, then the bounded change induced by shifting the window and perturbing the newly appended sample may be insufficient to restore the correct classification.

However, these sufficient conditions do not eliminate the need for a universal argument when certifying non-extendability in the general case:
unless such bounds are tight, proving that \emph{no} admissible perturbation yields an adversarial \(\hat X_{i+1}\) still requires solving a robustness verification problem.
\end{remark}

\end{document}